\documentclass[journal]{IEEEtran}
\usepackage{latexsym,amssymb}
\usepackage{mathrsfs}
\usepackage{epsfig}
\usepackage{graphicx}
\usepackage{color}
\usepackage{enumerate}

\newtheorem{theorem}{Theorem}

\newtheorem{proposition}{Proposition}
\newtheorem{assumption}{Assumption}
\newtheorem{lemma}{Lemma}

\newenvironment{proof}{\medskip\noindent{\it Proof. }}{ \medskip}
\newtheorem{remark}{Remark}

\def\begequarr{\begin{eqnarray}}
\def\endequarr{\end{eqnarray}}
\def\begequarrs{\begin{eqnarray*}}
\def\endequarrs{\end{eqnarray*}}
\def\begequ{\begin{equation}}
\def\endequ{\end{equation}}
\def\begequs{\begin{equation*}}
\def\endequs{\end{equation*}}
\def\begite{\begin{itemize}}
\def\endite{\end{itemize}}

\def\begcen{\begin{center}}
\def\endcen{\end{center}}
\def\begrem{\begin{remark}\rm}
\def\endrem{\end{remark}}
\def\ba{\begin{array}}
\def\ea{\end{array}}

\def\diag{\mbox{diag}}

\def\col{ \mbox{col}\; }
\def\blkdiag{\mbox{blkdiag}\;}

\def\dst{\displaystyle}

\def\beeq#1{\begin{equation}{#1}\end{equation}}
\def\qmx#1{\begin{pmatrix}{#1}\end{pmatrix}}

\begin{document}
\title{\huge Robust Output Feedback Stabilization of MIMO Invertible Nonlinear Systems with Output-Dependent Multipliers (Extended Version)}
\author{Lei Wang, Christopher M. Kellett
\thanks{*This work was supported by the Australian Research Council under ARC-DP160102138.}
\thanks{L. Wang is with Australian Center for Field Robotics, The University of Sydney, Australia. C. M. Kellett is with  School of Engineering, Australian National University, Australia.
 (E-mail: lei.wang2@sydney.edu.au; chris.kellett@anu.edu.au)}%
\thanks{This work was done while L. Wang and C. M. Kellett were with Faculty of Engineering and Built Environment, University of Newcastle, Australia.}%
        }
\date{}
\maketitle

\begin{abstract}
This note studies  the robust output feedback stabilization problem for  multi-input multi-output invertible nonlinear systems with output-dependent multipliers. An ``ideal" state feedback is first designed under certain mild assumptions. Then, a set of extended low-power high-gain observers is systematically designed, providing a complete estimation of the ``ideal" feedback law. This yields a robust output feedback stabilizer such that the origin of the closed-loop system is semiglobally asymptotically stable, while improving the numerical implementation  with the power of high-gain parameters up to 2.

\end{abstract}


\section{Introduction}
The problem of output feedback stabilization to the zero equilibrium for nonlinear systems is a fundamental problem in the field of systems and control. Several methodologies  have been developed such as high-gain observer-based, backstepping-based, and passivity-based control \cite{Isidori(1999),Khalil(book2002)}, that differ in the kind of system structure (normal form or lower triangular form), and in assumptions on the internal stability (input-output stability or output-to-state stability).

For single-input single-output (SISO) nonlinear systems, particular attention has been devoted to systems having a \emph{normal form}, for which a high-gain observer (HGO) can be employed to derive an output feedback stabilizer, providing asymptotic/practical stability in a semiglobal sense. Along this line, an extended high-gain observer (EHGO)-based approach is developed in \cite{Freidovich&Khalil(TAC2008)}, where an ``ideal" state feedback, consisting of a linear stabilizing term and a term to cancel the undesired terms (referred to as  a perturbation term uniformly), is first designed. A robust output feedback stabilizer can then be designed by using an EHGO to not only estimate states in the stabilizing term, but also provide a \emph{partial} estimation to the perturbation term. As a result, one can recover an ``ideal" system performance obtained by the ideal state feedback.
The EHGO technique has been extended to multi-input multi-output (MIMO) nonlinear systems with a well-defined vector relative degree \cite{WangSCL, Guo&Zhao(2013)}, for which a static state feedback law for feedback linearization can be designed to decouple all input-output channels. However, the class of systems considered in  \cite{WangSCL,Guo&Zhao(2013)} is a very particular one, while the stabilization of more general classes of MIMO nonlinear systems is generally nontrivial and cannot be achieved via direct extensions of  SISO results.


Recently,  several authors have studied a general class of MIMO nonlinear systems \cite{WangTAC,Wang&Isidori(AUT2017),Wang&Isidori(TAC2017),Isidori(2017),Wu&Isidori(2019)}, referred to as invertible MIMO nonlinear systems \cite{Hirschorn, Singh},  for which a vector relative degree is not necessary. A significant feature of these systems is that the input-output behavior cannot be fully decoupled by a static feedback linearization law \cite{Isidori(2017)}, due to the presence of nonzero ``multipliers" when running the Structure Algorithm \cite{Singh,Isidori(2017)} to derive a multivariable normal form.
In \cite{WangTAC}, with an input-output linearizable assumption (i.e., having constant multipliers), by defining a ``virtual" output as a linear combination of actual outputs and their derivatives, the invertible systems can be transformed to an ``intermediate" form with a unitary vector relative degree, for which the feedback stabilization problem can be solved. This linearizable assumption is later relaxed in \cite{Wang&Isidori(TAC2017)} by permitting state-dependent ``multipliers" in a special structure such that a \emph{dynamical} feedback linearization can be used, but at the price of requiring a \emph{trivial} zero dynamics\footnote{The zero dynamics is said to be trivial if the constraint that the outputs are zero implies that the states are zero. Otherwise, we say that it is nontrivial.}. To apply the EHGO technique to robustify the stabilizer while recovering the dynamical feedback linearizing performance, \cite{Wu&Isidori(2019)} has further proposed a recursive design method  for the same class of invertible MIMO nonlinear systems with a lower-triangular high-frequency gain matrix. In spite of these impressive results, we note that the output feedback stabilization of invertible MIMO nonlinear systems with non-constant multipliers and nontrivial zero dynamics is still an open problem.

On the other hand,  in \cite{Freidovich&Khalil(TAC2008),WangSCL,Wu&Isidori(2019)} the maximum power of the high-gain parameter increases as the number of states increases, which in practice may create numerical implementation problems when the dimension of the system to be estimated is very large. To solve this problem, the low-power technique in \cite{Daniele&Lorenzo,Khalil(2017book)} can be employed. However, the combination with the low-power technique is nontrivial, particularly for invertible MIMO nonlinear systems.

Motivated by the previous context, this  note studies the problem of robust output feedback stabilization for  MIMO invertible nonlinear systems with output-dependent multipliers and nontrivial zero dynamics.  Compared to the relevant works \cite{Wang&Isidori(TAC2017),Wu&Isidori(2019)}, a nontrivial zero dynamics is permitted  and a lower triangular high-frequency gain matrix is not needed, though this note requires a stronger condition on multipliers. To the best knowledge of the authors, currently available approaches cannot be used to solve the considered problem.

In this paper, assuming that all states are accessible, an ideal state feedback law is first designed, rendering an asymptotically stable closed-loop system under a strongly minimum-phase condition. Taking advantage of both EHGO and low-power techniques, a set of extended low-power high-gain observers (ELPHGOs) is systematically designed, providing a \emph{complete} estimation of the ideal state feedback law. This in turn yields a robust output feedback stabilizer such that the origin of the resulting closed-loop system is semiglobally asymptotically stable.
Meanwhile, each EHGO has the power of its high-gain parameter only up to 2, which to some extent solves the numerical implementation problem.
It is worth noting that our ELPHGOs are designed with an estimation of the \emph{entire} perturbation term to achieve a \emph{complete} estimation of the ideal state feedback law, which is different from the \emph{partial} estimation in  \cite{Freidovich&Khalil(TAC2008),WangSCL,Wu&Isidori(2019)}. As further discussed in Remark \ref{remark-3}, this complete estimation in turn adds extra difficulties and challenges to the stability analysis.


{\bf Notations: } $|\cdot|$ denotes the standard Euclidean norm. $\otimes$ denotes the Kronecker product. A continuous function $\alpha:\mathbb{R}_+:=[0,\infty)\rightarrow\mathbb{R}_+$ is said to be of class $\mathcal{K}$ if $\alpha$ is strictly increasing and $\alpha(0)=0$, and of class $\mathcal{K}_{\infty}$ if it is also unbounded.  For any positive integer $d$, $\mathbf{0}_d$ denotes a $d\times 1$ vector, whose entries  are all zero, and $(A_d,B_d,C_d)$ is used to denote the matrix triplet in the prime form. Namely, $A_d$  denotes a shift matrix of dimension $d\times d$, $B_d=\qmx{0&\cdots&0&1}^{\top}\in\mathbb{R}^d$, and $C_d=\qmx{1&0&\cdots&0}\in\mathbb{R}^{1\times d}$. For any $x_i\in\mathbb{R}^{d_i}$, $i=1,\ldots,n$, we denote $\col(x_1,\ldots,x_n)$ as a vector in $\mathbb{R}^{\sum_{i=1}^nd_i}$ by concatenating all $x_i$'s in order. We denote $\mbox{satv}_l(s)$ as  a smooth vector-valued saturation function with saturation level $l$: $\mbox{satv}_l(s)=s$ if $|s|\leq l$; $0<\frac{d\,\mbox{satv}_l(s)}{ds}<1$ for all $|s|>l$; and $\lim_{s\rightarrow\infty}\mbox{satv}_l(s)=l+\epsilon_0$ with $0<\epsilon_0\ll1$.
 For convenience, $\nabla\mbox{satv}_l$ denotes the Jacobian matrix of $\mbox{satv}_l(\cdot)$.

\section{Preliminaries}
\label{sec-problemstatement}


{
Consider invertible MIMO nonlinear systems of the form
\vspace{-0.5em}
\beeq{\label{partialform}
\ba{l}
\dot x_0 = f_0(x_0,\xi,u)\,\\
\left\{
\ba{l}
{\dot \xi_1 = A_{r_1} \xi_1 + B_{r_1}[a_1(x) + b_1(x)u]}\\
y_1 = \xi_{1,1}
\ea
 \right.\\
\left\{\ba{l}
{\dot \xi_{k} = A_{r_k} \xi_k + B_{r_k} [a_k(x) + b_k(x)u]} \,\\ \qquad {+ \sum_{i=1}^{k-1} M_{k}^i(y) [a_i(x) + b_i(x)u]}\\
y_k = \xi_{k,1},\qquad 2\leq k\leq m
\ea\right.
\ea}

\vspace{-0.5em}
\noindent where internal states $x_0\in\mathbb{R}^{n_0}$ and partial states $\xi=\col(\xi_1,\ldots,\xi_m)\in\mathbb{R}^{\bf r}$ with $\xi_i=\col(\xi_{i,1},\ldots,\xi_{i,r_i})\in\mathbb{R}^{r_i}$ and $\mathbf{r}=\sum_{i=1}^mr_i$, output $y=\col(y_1,\ldots,y_m)$, control $u\in\mathbb{R}^m$,
and multiplier vectors $M_{k}^i: \mathbb{R}^m\rightarrow\mathbb{R}^{r_k}$  are defined by
\vspace{-0.5em}
\begin{equation}\label{eq:multiplier}
M_{k}^i(y) = \qmx{\mathbf{0}_{r_i-1}^{\top} & \delta_{k,r_i+1}^i(y) & \cdots & \delta_{k,r_k}^i(y) & 0}^\top\,.
\end{equation}

\vspace{-0.5em}
\noindent
For convenience, we let
\vspace{-0.5em}
\beeq{\label{eq:x}
x=\col(x_0,\xi)\in\mathbb{R}^n,\quad \mbox{ with $n=n_0+\mathbf{r}$.}
}

\vspace{-0.5em}
\noindent
Let $a(x)\in \mathbb{R}^{m}$ be a vector with the $i$-th entry $a_i(x)$, and $b(x)\in \mathbb{R}^{m\times m}$ be a matrix with  the $i$-th row $b_i(x)$. Throughout this paper, we suppose all mappings $f_0,a_i,b_i,\delta_{k,i+1}^s$ in (\ref{partialform}) are sufficiently smooth, and $a(0)=0$ and $b(x)$ is invertible for all $x\in\mathbb{R}^n$.
In this setting, this paper is interested in the problem of semiglobal asymptotic stabilization of system (\ref{partialform}) via output feedback.
As in \cite{Isidori(2017)}, we assume system (\ref{partialform}) is strongly---and also locally exponentially---minimum-phase.

\begin{assumption}\label{ass-MPP}
There exists an ISS Lyapunov function $V_0(x_0)$  for the $x_0$-subsystem with $\xi$ as an input, uniformly in $u$, such that $V_0(x_0)$ is positive definite and radially unbounded, and along the $x_0$-subsystem in (\ref{partialform}), for every $(x_0(0),\xi(0))\in\mathbb{R}^n$ and all $u\in\mathbb{R}^m$,
\vspace{-0.5em}
\beeq{\ba{l}\label{ineq-iss}
\dot V_0(x_0) \leq -\beta_1(|x_0|) + \beta_2(|\xi|)\,
\ea}

\vspace{-0.5em}
\noindent
holds for some $\beta_1,\beta_2\in\mathcal{K}_{\infty}$. The origin of system $\dot x_0 = f_0(x_0,0,u)$ is locally exponentially stable, uniformly in $u$.
\end{assumption}

If $x$ is available for feedback and if the functions $a(\cdot)$ and $b(\cdot)$ are known, we can design an ``ideal" control law
\vspace{-0.5em}
\beeq{\label{ideal-u}
u^\ast = b^{-1}(x)[-a(x) + v]
}

\vspace{-0.5em}
\noindent
with the residual control $v=\col(v_1,\ldots,v_m)$.

This ideal control reduces the input-output model of (\ref{partialform}) to
\vspace{-0.5em}
\beeq{\label{eq:xi-ideal}\ba{l}
\dot \xi_1 = A_{r_1} \xi_1 + B_{r_1} v_1\,\\
\dot \xi_k = A_{r_k} \xi_k + B_{r_k} v_k + \sum_{i=1}^{k-1} M_{k}^i(y) v_i\,,  2\leq k\leq m\,.
\ea}

\vspace{-0.5em}
\noindent
Regarding the design of $v$ stabilizing (\ref{eq:xi-ideal}), several approaches can be applied, such as high-gain method \cite{Isidori(1999)} for a semiglobal stabilizer.
In this paper, since our focus is not on the design of $v$ for (\ref{eq:xi-ideal}), and to ease the subsequent analysis, we additionally make the following assumption.
\begin{assumption}\label{ass-delta}
  All multipliers $\delta_{k,j}^i(y)$ in (\ref{eq:multiplier}) are bounded for all $y\in\mathbb{R}^m$.
\end{assumption}
This then motivates us to select $K_k\in\mathbb{R}^{1\times r_k}$  such that $A_{r_k}-B_{r_k}K_k$ is Hurwitz, and design $v$ as
\vspace{-0.5em}
\beeq{\label{v}
v_k = -K_k\,\xi_k \,, \quad 1\leq k\leq m\,.
}

\vspace{-0.5em}
Thus, by  Assumption \ref{ass-delta}, it can be easily seen that the origin of the $\xi_1$-subsystem in (\ref{eq:xi-ideal}) is globally exponentially stable (GES), and  the $\xi_2$-subsystem is input-to-state stable (ISS) with respect to input $\xi_1$ with a linear gain. Hence, the origin of the $(\xi_1,\xi_2)$-subsystems turns out GES by the  small-gain theorem \cite{Kellett2015}. Similarly, we can further  consider the $\xi_k$-subsystem recursively for $k=3,\ldots,m$ and eventually obtain that the origin of (\ref{eq:xi-ideal}) is GES with a quadratic Lyapunov function. This, with (\ref{ineq-iss}), implies that the origin of system (\ref{partialform}) with  (\ref{ideal-u}), (\ref{v}) is globally asymptotically stable using the standard small-gain theorem \cite{Ito2002,Kellett2015}, permitting a constructive Lyapunov function $V_x(x)$ for (\ref{partialform}) and an $\alpha_{x}\in \mathcal{K}_\infty$ such that $V_x(x)$ is positive definite and radially unbounded, and
\vspace{-0.5em}
\beeq{\label{V_x}
\dot V_x \leq -\alpha_{x}(|x|)\,.
}

\vspace{-0.5em}
We note that the feedback law (\ref{ideal-u})-(\ref{v}) is not implementable due to inaccessibility of the full state $x$.  Motivated by this, this note develops a new set of high-gain observers driven only by the output $y$, which provides an estimate of the  controller  (\ref{ideal-u})-(\ref{v}), stabilizing the origin of system (\ref{partialform}) in a semiglobal sense.

\begin{remark}
The considered class of  systems (\ref{partialform}) can be regarded as a particular case of the \emph{multivariable normal form} \cite{Isidori(1999)} with output-dependent multipliers $\delta_{k,i+1}^s$, and can be easily verified to be invertible in the sense of \cite{Singh}. The derivations of such form (\ref{partialform}) can follow the Structure Algorithm \cite{Isidori(2017)}.  We note that the stabilization problem of invertible MIMO system (\ref{partialform}), to the best knowledge of authors, has not been studied yet and cannot be solved by the existing approaches \cite{WangTAC,Wang&Isidori(AUT2017),Wang&Isidori(TAC2017),Wu&Isidori(2019)}.
Compared to the invertible MIMO nonlinear systems  having constant multipliers in \cite{WangTAC,Wang&Isidori(AUT2017)}, the multipliers $\delta_{k,i+1}^s(y)$ in (\ref{partialform}) are output-dependent, for which the ``intermediate" form with the vector relative degree $\{1,1,\ldots,1\}$ and strongly minimum-phase property cannot be obtained from ``virtual" outputs that are defined by a linear function of actual outputs and their derivatives. As a result, the recursive observer design approach in \cite{WangTAC,Wang&Isidori(AUT2017)} cannot be applied.  In contrast with \cite{Wang&Isidori(TAC2017),Wu&Isidori(2019)}, our model (\ref{partialform}) permits the existence of the zero dynamics (or $x_0$-dynamics in (\ref{partialform})) and the high-frequency gain matrix $b(x)$ to be in a general, instead of lower-triangular, structure. $\square$
\end{remark}

\begin{remark}
  In Assumption \ref{ass-MPP}, (\ref{ineq-iss}) characterizes that the system (\ref{partialform}) is weakly uniformly 0-detectable of order $\{r_1,\ldots,r_m\}$ in the sense of \cite{Liberzon(TAC2003)}. We note that in the present semiglobal setting, Assumption \ref{ass-delta} is not necessary and (\ref{v}) can be replaced by a high-gain feedback law \cite{Isidori(1999)}. $\square$
\end{remark}

}


\section{Observer and Control Design}

Let $\mathcal{C}_x\subset\mathbb{R}^n$ be any compact set, and $c>0$ be such that
\[
\mathcal{C}_x\subset \Omega_c: = \{x\in\mathbb{R}^n: V_x(x)\leq c\}\,
\]
where $V_x(x)$ is defined in (\ref{V_x}).
As in \cite{WangSCL}, we assume that the high-frequency gain matrix $b(x)$ satisfies the property below.
\begin{assumption}\label{ass-b}
For all $x\in\Omega_{c+1}$, there exist a constant nonsingular matrix $\widehat {\mathsf{B}} \in \mathbb{R}^{m\times m}$ and a number $0<\mu_0<1$ such that
\beeq{\label{bproperty}
|\Delta_b(x)| \le \mu_0 \,,\quad \mbox{with  $\Delta_b(x):=(b(x)-\widehat {\mathsf{B}})\widehat {\mathsf{B}}^{-1}$.}
}
\end{assumption}

Define the perturbation term \footnote{For readability, the arguments $(x,u)$ of $\sigma$ will be omitted occasionally.}
\beeq{\label{sigma}
\sigma(x,u) := a(x) + [b(x)-\widehat {\mathsf{B}}]u\,
}
and let $K=\mbox{blkdiag}(K_1,\ldots,K_m)$. {By simple calculations, it can be seen that the equation
\beeq{\label{u-ast}
u = -\widehat {\mathsf{B}}^{-1}(\sigma(x,u) + K\xi)\,
}
has the unique solution $u=u^\ast$, with $u^\ast$ denoting the ideal control law (\ref{ideal-u})-(\ref{v}).

In view of this, consider the control law $u = -\widehat {\mathsf{B}}^{-1}(\hat\sigma+ K\hat\xi)$ using $\hat\xi, \hat\sigma$  as estimates of the partial state $\xi$ and the perturbation $\sigma$, which are provided by some observer. If $\hat\sigma$ and $\hat\xi$ are asymptotic estimates of $\sigma$ and $\xi$, respectively, then it is clear from (\ref{u-ast}) that the  control law $u = -\widehat {\mathsf{B}}^{-1}(\hat\sigma+ K\hat\xi)$ is  an asymptotic estimate of the ideal control $u^\ast$.
This intuition motivates the subsequent output feedback controller design.}

With (\ref{sigma}), by letting $\widehat {\mathsf{B}}_i$ be the $i$-th row of matrix $\widehat {\mathsf{B}}$ and $\sigma_i$ be the $i$-th entry of $\sigma$, we can rewrite the input-output model of  (\ref{partialform}) as
\begin{equation}\label{eq:11}\ba{l}
 \dot \xi_1 = A_{r_1}\xi_1 + B_{r_1}[\sigma_1 + \widehat {\mathsf{B}}_1 u]\,\\
 \dot \xi_k = A_{r_k}\xi_k + \sum_{i=1}^{k-1} M_{k}^i(y)[\sigma_i + \widehat {\mathsf{B}}_i u] \\ \qquad + B_{r_k}[\sigma_k + \widehat {\mathsf{B}}_k u]\,, \quad 2\leq k\leq m\,.
\ea\end{equation}
Denote $M_{k,j}^i(y)$ as the $j$-th entry of vector $M_{k}^i(y)$ and note that $M_{k,j}^i=0$ for $j=1,\ldots,r_i-1,r_k$ and $M_{k,j}^i(y)=\delta_{k,j+1}^i(y)$ for $j=r_i,r_i+1,\ldots,r_k-1$. Let $\widehat {\mathsf{M}}_{k,j}^i(y)=\big[M_{k,j}^i(y)\,\, M_{k,j+1}^i(y)\big]^\top\in\mathbb{R}^2$ for $1\leq j\leq r_k-1$ and $1\leq i\leq k-1\leq m-1$. For (\ref{eq:11}), we propose a set of high-gain observers of the form {
\beeq{\label{ELPHGO-1}\ba{l}
\dot \eta_{1i} = A_2\eta_{1i} + D_2\eta_{1,i+1} + \Lambda_2(\ell_1)\Gamma_{1i}(\eta_{1,i-1}^2-\eta_{1i}^1)\,,\\
\qquad 1\leq i \leq r_1-2\,\\
\dot \eta_{1,r_1-1} = A_2\eta_{1,r_1-1} + D_2\eta_{1,r_1} + B_2\widehat {\mathsf{B}}_1 u \,\\ \qquad + \Lambda_2(\ell_1)\Gamma_{1,r_1-1}(\eta_{1,r_1-2}^2-\eta_{1,r_1-1}^1)\,\\
\dot \eta_{1,r_1} = A_2\eta_{1,r_1} + C_2^\top\widehat {\mathsf{B}}_1 u + \Lambda_2(\ell_1)\Gamma_{1,r_1}(\eta_{1,r_1-1}^2-\eta_{1,r_1}^1)\,\\
\ea}
and for $k=2,\ldots,m$,
\beeq{\label{ELPHGO-2}\ba{l}
\dot\eta_{ki} = A_2\eta_{ki}+D_2\eta_{k,i+1}+\sum_{j=1}^{k-1}\widehat {\mathsf{M}}_{k,i}^j(y)(\eta_{j,r_j}^2+ \widehat {\mathsf{B}}_j u)\\ \qquad + \Lambda_2(\ell_k)\Gamma_{k,i}(\eta_{k,i-1}^2-\eta_{k,i}^1)\,,\quad 1\leq i\leq r_k-2\\
\dot\eta_{k,i} = A_2\eta_{k,i}+D_2\eta_{k,i+1}+\sum_{j=1}^{k-1}\widehat {\mathsf{M}}_{k,i}^j(y)(\eta_{j,r_j}^2+ \widehat {\mathsf{B}}_j u)\\ \qquad +B_2\widehat {\mathsf{B}}_k u + \Lambda_2(\ell_k)\Gamma_{k,i}(\eta_{k,i-1}^2-\eta_{k,i}^1)\,,\quad i=r_k-1\\
\dot\eta_{k,r_k} = A_2\eta_{k,r_k} + C_2^\top\widehat {\mathsf{B}}_k u + \Lambda_2(\ell_k)\Gamma_{k,r_k}(\eta_{k,r_k-1}^2-\eta_{k,r_k}^1)
\ea}
where for $1\leq k\leq m$, $\eta_{k,0}^2=y_k$,}  $\eta_{k,j}=\col(\eta_{k,j}^1,\eta_{k,j}^2)\in\mathbb{R}^2$,
\[\ba{l}
D_2 = \mbox{diag}(0, 1)\in\mathbb{R}^{2\times2}\,,\quad \Lambda_2(\ell_k)=\mbox{diag}(\ell_k ,(\ell_k)^2)\in\mathbb{R}^{2\times2}\,
\ea\]
with $\ell_k$ being high-gain parameters to be determined, and $\Gamma_{k,i}=\col(\gamma_{k,i}^1, \gamma_{k,i}^2)$, $ i=1,\ldots,r_k$ are such that the matrix\footnote{Existence of $\gamma_{k,i}^1,\gamma_{k,i}^2>0$ leading to a Hurwitz $F_k$ is always guaranteed (see e.g. \cite{Daniele&Lorenzo,Wang&Marconi(AUT2017)} for explicit proofs and design methods).}
\[F_k = \qmx{F_{k,1} & D_2 & \cdots  & 0 & 0  \cr
           \Gamma_{k,2}B_2^{\top} & F_{k,2} &   \ddots &  0 & 0 \cr
           \cdots & \ddots & \ddots &  \ddots &   \vdots \cr
           0 & 0 &   \ddots & F_{k,r_k-1} & D_2 \cr
           0 & 0 &   \cdots &  \Gamma_{k,r_k}B_2^{\top} & F_{k,r_k} \cr
            }
\]
is Hurwitz,
with  $F_{k,i}=A_2-\Gamma_{k,i}C_2$.


Let $\hat\xi_k$ and $\hat\sigma_k$ denote the estimates of $\xi_k$ and $\sigma_k$, respectively, the expressions of which are given by
\beeq{\label{eq:estimates}
\hat\xi_k=(I_{r_k}\otimes C_2) \eta_k\,,\quad \hat\sigma_k=\eta_{k,r_k}^2\,,\quad k=1,\ldots,m\,.
}


Letting $\eta=\col(\eta_1,\ldots,\eta_m)\in\mathbb{R}^{2{\bf r}}$ with $\eta_k=\col(\eta_{k,1},\ldots,\eta_{k,r_k})\in\mathbb{R}^{2r_k}$, $\hat\xi=\col(\hat\xi_1,\ldots,\hat\xi_m)\in\mathbb{R}^{{\bf r}}$ and $\hat\sigma=\col(\hat\sigma_1,\ldots,\hat\sigma_m)\in\mathbb{R}^{m}$, instead of the ideal feedback control (\ref{ideal-u})-(\ref{v}), we propose an implementable feedback law as
\beeq{\label{actual-u}
u = -\widehat {\mathsf{B}}^{-1}\mbox{satv}_{l}(\hat\sigma + K\hat\xi)
}
where  the vector-valued saturation function $\mbox{satv}_{l}(\cdot)$ is introduced with
{ saturation level $l$ designed as
\beeq{\label{eq:l}
l \geq \sup_{x\in\Omega_{c+1}}\frac{1}{1-\mu_0}\left|a(x)+K\xi\right| + 1\,.
}
}
\vspace{-0.8em}
{
\begin{remark}
  With $\hat\xi, \hat\sigma$ being estimates of  $\xi$ and $\sigma$, respectively,   (\ref{actual-u}) can be regarded as an estimate of (\ref{u-ast}), by ignoring the effect of the saturation function $\mbox{satv}_{l}$ that is introduced to avoid the finite-time escape phenomenon \cite{Khalil1992}. If $(\hat\xi(t), \hat\sigma(t))$ asymptotically converge to $(\xi(t), \sigma(t))$, then the $u$ in (\ref{actual-u}) will converge to the ideal control  $u^\ast$ in (\ref{ideal-u})-(\ref{v}) as $u^\ast$ is the unique solution of (\ref{u-ast}). $\square$
\end{remark}
}
\begin{remark}
{
  The observer (\ref{ELPHGO-1})-(\ref{ELPHGO-2}) is comprised of $m$ high-gain observers, and  belongs to the class of extended observers \cite{Freidovich&Khalil(TAC2008),WangSCL}. Conventional  observers \cite{Daniele&Lorenzo,Wang&Isidori(TAC2017)}  estimate only the partial state, while (\ref{ELPHGO-1})-(\ref{ELPHGO-2}) provides estimates $\hat\xi_i,\hat\sigma_i$ for both the partial states $\xi_i$ and the perturbation term $\sigma_i$, $i=1,\ldots,m$, and can recover the performance of the ideal state feedback law (see e.g., \cite{Freidovich&Khalil(TAC2008),Wu&Isidori(2019)}).
  Besides, with the use of the extended high-gain observer to estimate both $\xi_i$ and $\sigma_i$, we can match each unavailable quantity in (\ref{eq:11}) with an estimate in the observer design, leading to a matched observer (\ref{ELPHGO-1})-(\ref{ELPHGO-2}). This eventually allows us to stabilize the resulting estimation error dynamics at an arbitrarily small neighborhood of the origin by adjusting the high-gain parameters, as in the forthcoming Proposition 1.
  }
  On the other hand, the design of (\ref{ELPHGO-1})-(\ref{ELPHGO-2}) also utilizes the low-power technique developed in \cite{Daniele&Lorenzo} for the purpose of solving the numerical implementation problem when $r_k$ is very large. As one can see, the high-gain parameter $\ell_k$ of each observer is powered up to only 2, rather than $r_k+1$ as in \cite{Freidovich&Khalil(TAC2008),WangSCL}, although the dimension of the observer increases to $2r_k$. $\square$
\end{remark}

\section{Stability Analysis}




\subsection{Change of Coordinates}
The aim of this subsection is to derive the estimation error dynamics, whose stability will be analyzed later.

%

For $1\leq k\leq m$, define the re-scaled estimation errors as
\beeq{\label{errors}
\ba{l}
\left\{
\ba{l}
\tilde\eta_{k,1}^1 = {(\ell_k)}^{r_k}(y_k-\eta_{k,1}^1)\,\\
\tilde\eta_{k,1}^2 = {(\ell_k)}^{r_k-1}(\xi_{k,2}-\eta_{k,1}^2)\,\\
\ea
\right.\\
\left\{
\ba{l}
\tilde\eta_{k,i}^1 = {(\ell_k)}^{r_k-i+1}(\xi_{k,i}-\eta_{k,i}^1)\\
\tilde\eta_{k,i}^2 = {(\ell_k)}^{r_k-i}(\xi_{k,i+1}-\eta_{k,i}^2)\\
\ea
\right., 1\leq i\leq r_k-1\\
\left\{
\ba{l}
\tilde\eta_{k,,r_k}^1 = {\ell_k}(\xi_{k,r_k}-\eta_{k,r_k}^1)\,\\
\tilde\eta_{k,,r_k}^2 = \sigma_k -\eta_{k,r_k}^2\,\\
\ea
\right.\\
\ea
}
with $\sigma_k$ being the $k$-th element of vector $\sigma$ defined in (\ref{sigma}).

Let  $\tilde\eta_{k,j}=\col(\tilde\eta_{k,j}^1,\tilde\eta_{k,j}^2)$, $\tilde\eta_k=\col(\tilde\eta_{k,1},\ldots,\tilde\eta_{k,r_k})$, for $1\leq k\leq m$ and $1\leq j \leq r_k$ and $\tilde\eta=\col(\tilde\eta_1,\ldots,\tilde\eta_m)$. Let
\beeq{\label{eq:t-sigma}
\tilde\sigma=\col(\tilde\eta_{1,r_1}^2,\tilde\eta_{2,r_2}^2,\ldots,\tilde\eta_{m,r_m}^2)\,.
}

\begin{remark}\label{remark-3}
From the bottom equation of (\ref{errors}), it can be seen that the extra estimate $\eta_{k,r_k}^2$ is used to estimate the \emph{entire} perturbation $\sigma_k$, which is motivated by \cite{Praly&Jiang(CDC1998),Praly&Jiang(1998)} and different from \cite{Freidovich&Khalil(TAC2008),WangSCL}. {As an example,  consider a simple case of system (\ref{partialform}) with $r_1=1$ and $r_2\geq 2$. Let $\bar\Lambda_{\ell_2}=\diag(\ell_2,\ldots,(\ell_2)^{r_2})$ and $\bar\Gamma_{2}=\col(\gamma_{21},\ldots,\gamma_{2,r_2})$. Following  \cite{Freidovich&Khalil(TAC2008),WangSCL}, we  can design  the extended high-gain observer as
\[\ba{l}
\dot {\hat\xi}_{1} = \hat{\sigma}_{1} + \widehat {\mathsf{B}}_1 u+ \ell_1\gamma_{11}(y_1-\hat\xi_{1})\,,\quad
\dot {\hat{\sigma}}_{1} =  (\ell_1)^2\gamma_{12}(y_1-\hat\xi_{1})\\
\dot {\hat\xi}_{2} = A_{r_2}\hat\xi_{2} + B_{r_2}(\hat{\sigma}_{2} + \widehat {\mathsf{B}}_2 u) + M_{2}^1(y)(\hat{\sigma}_{1} + \widehat {\mathsf{B}}_1 u)\,\\ \qquad  + \bar\Lambda_{\ell_2}\bar\Gamma_{2}(y_2-C_{r_2}\hat\xi_{2})\\
\dot {\hat{\sigma}}_{2} =  (\ell_2)^{r_2+1}\gamma_{2,r_2+1}(y_2-C_{r_2}\hat\xi_{2})\,,
\ea\]
with the extra estimates ${\hat{\sigma}}=\col({\hat{\sigma}}_{1},{\hat{\sigma}}_{2})$, and define the re-scaled estimation errors as
\[\ba{l}
\tilde \xi_{1} := \ell_1(\xi_{1}-\hat\xi_{1})\,,\quad \tilde \xi_{2}: = \mbox{diag}\big((\ell_2)^{r_2},\ldots,\ell_2\big)(\xi_{2}-\hat\xi_{2}),\\
\breve\sigma_i := \bar\sigma_i - \hat{\sigma}_i, \qquad i=1,2
\ea\]
with
\[\ba{l}
\bar\sigma_i = a_i(x) - [b_i(x)-\widehat {\mathsf{B}}_i]\widehat {\mathsf{B}}^{-1}\mbox{satv}_{l}(\hat{\sigma} + K\xi)\,\\
= \sigma_i - [b_i(x)-\widehat {\mathsf{B}}_i]\widehat {\mathsf{B}}^{-1}\big(\mbox{satv}_{l}(\hat{\sigma} + K\xi)-\mbox{satv}_{l}(\hat{\sigma} + K\hat\xi)\big)\,.
\ea\]
From the above definition of $\breve\sigma_i$, we note that the extra estimates ${\hat{\sigma}}_{i}$ are used to estimate $\bar\sigma_i$, which is different from the perturbation terms $\sigma_i=a_i(x) - [b_i(x)-\widehat {\mathsf{B}}_i]\widehat {\mathsf{B}}^{-1}\mbox{satv}_{l}(\hat\sigma + K\hat\xi)$ due to the gap between $\xi$ and its estimate $\hat\xi$. Thus, the extra estimates provided by the extended high-gain observer in \cite{Freidovich&Khalil(TAC2008),WangSCL} are  \emph{partial} estimates of the perturbations, in contrast to our observers which estimate the \emph{entire} perturbations by (\ref{errors}). Besides, by letting $\tilde\xi_2=\col(\tilde\xi_{21},\ldots,\tilde\xi_{2,r_2})$, the resulting $\tilde \xi_{2k}$-dynamics, $k=1,\ldots,r_2-1$ is given by
\[
\dot {\tilde\xi}_{2k} = \ell_2(\tilde\xi_{2,k+1}-\gamma_{2k}\tilde\xi_{21} ) + (\ell_2)^{r_2+1-k}\delta_{2,k+1}^1(y)\breve\sigma_1 + \mathcal{J}_k
\]
for some appropriate term $\mathcal{J}_k$, vanishing as the gap between $\sigma_1$ and $\bar\sigma_1$ vanishes and satisfying that $|\mathcal{J}_k|$ is lower-bounded by a quantity with $(\ell_2)^{j}|\tilde\xi_{2,k+j}|$, $j=1,\ldots,r_2-k$. As a result, there is no guarantee that the effect of $\mathcal{J}_k$  can be  dominated by adjusting the high-gain parameters when computing the time-derivative of the Lyapunov function for the re-scaled estimation error dynamics. This eventually prevents analyzing  closed-loop stability. In this paper, this $\mathcal{J}_k$ does not exist since $\eta_{i,r_i}^2$ or $\hat\sigma_i$ is used to estimate the \emph{entire} perturbation $\sigma_i$, leading to a \emph{complete} estimation of $u^\ast$. As will be shown later, this \emph{complete} estimation allows us to analyze  closed-loop stability by appropriately designing the high-gain parameters.}
However, as the perturbations $\sigma$ to be estimated depend on the control input $u$, the corresponding stability analysis is  non-trivial. $\square$
\end{remark}

Now we proceed to rewrite the closed-loop system (\ref{partialform})-(\ref{ELPHGO-1})-(\ref{ELPHGO-2})-(\ref{actual-u}) in coordinates $(x,\tilde\eta)$. With  (\ref{errors}), we rewrite (\ref{actual-u}) as
\beeq{\label{u-error}
u = -\widehat {\mathsf{B}}^{-1}\mbox{satv}_{l}(\sigma(x,u) + K\xi - \tilde\sigma - K(\Lambda_{\ell}^{-1}\otimes C_2)\tilde\eta)
}
where $\Lambda_{\ell}=\mbox{blkdiag}(\Lambda_{\ell_1},\ldots,\Lambda_{\ell_m})$ with $\Lambda_{\ell_k}=\diag(\ell_k^{r_k},\ldots,\ell_k)$, and $\sigma$, as defined in (\ref{sigma}), depends on $x,u$.  The following lemma shows that the equation (\ref{u-error}) has the unique solution $u$.
\begin{lemma}\label{lemma-psiu}
Set $\psi(u) = u +\widehat {\mathsf{B}}^{-1}\mbox{satv}_{l}(\sigma(x,u) + K\xi - \tilde\sigma - K(\Lambda_{\ell}^{-1}\otimes C_2)\tilde\eta)$. Then, with Assumption \ref{ass-b},  there exists a unique solution of the equation $\psi(u)=0$ for all $x\in\Omega_{c+1}$.
\end{lemma}

{
With Assumption \ref{ass-b}, some simple calculations show that $\psi(\cdot)$ is proper (i.e., radially unbounded, namely $|\psi(u)|\rightarrow\infty$ as $|u|\rightarrow\infty$) and its Jacobian $\frac{\partial\psi(u)}{\partial u}$ is uniformly nonsingular for all $u\in\mathbb{R}^m$, which, according to Hadamard's Theorem \cite{Hadamard1906}, indicates that  the function $\psi$ has a globally defined inverse. This in turn proves  Lemma \ref{lemma-psiu}.} We omit the corresponding details. Then, recalling (\ref{eq:x}) and (\ref{eq:t-sigma}), the unique solution $u$ of (\ref{u-error}) is a function of $(x,\tilde\eta)$, and we denote it as $u=\pi(x,\tilde\eta)$.


Thus, we can rewrite (\ref{partialform}) with (\ref{actual-u}) in coordinates $(x,\tilde\eta)$ as
\beeq{\label{xi-3}
\ba{l}
\dot x_0 = f_0(x_0,\xi,\pi(x,\tilde\eta))\,\\
\dot \xi = (\mathbf{A} - \mathbf{B}(y)K)\xi + \mathbf{B}(y)[\phi(x,\tilde\eta)-\mbox{satv}_l\left(\phi(x,\tilde\eta) \right.\,\\ \qquad \qquad \left.- \tilde\sigma - K(\Lambda_{\ell}^{-1}\otimes C_2)\tilde\eta\right)]\\
\ea}
in which $\mathbf{A} = \mbox{blkdiag}(A_{r_1},\ldots,A_{r_m})$, and
\[
\mathbf{B}(y) = \qmx{B_{r_1} & 0 &\cdot & 0 \cr M_{2}^1(y) & B_{r_2} &\cdot & 0 \cr \cdot & \cdot &\cdot & \cdot \cr M_{m}^1(y) & M_{m}^2(y) &\cdot & B_{r_m} \cr}\,
\]
\beeq{
\phi(x,\tilde\eta)=K\xi+a(x)+(b(x)-\widehat {\mathsf{B}})\pi(x,\tilde\eta) \label{eq-phi}\,.
}

Taking the derivative of the estimation errors in (\ref{errors}) yields
{
\beeq{\label{tilde-eta-1}\ba{l}
\dot {\tilde\eta}_{1} = \ell_1 F_{1}{\tilde\eta}_{1}  + B_{2r_1} \dot\sigma_1  \,,
\ea}
}
and for $k=2,\ldots,m$,
{
\beeq{\label{tilde-eta-k}\ba{l}
\dot {\tilde\eta}_{k} =  \ell_k F_{k}{\tilde\eta}_{k} + \sum_{j=1}^{k-1}L_{kj}(\ell_k,y)B_{2r_j}^\top{\tilde\eta}_{j} + B_{2r_k} \dot\sigma_k\,
\ea}
}
with
\[
L_{kj}(\ell_k,y) = \qmx{\mathbf{0}_{2r_j-3}\cr (\ell_k)^{r_k-r_j+1}\delta_{k,r_j+1}^j(y)\mathbf{1}_2\cr\cdots\cr (\ell_k)^2\delta_{k,r_k}^j(y)\mathbf{1}_2\cr \mathbf{0}_3}\,.
\]

Putting all bottom equations of the vector equations (\ref{tilde-eta-1}) and (\ref{tilde-eta-k}) together, and recalling (\ref{eq:t-sigma}), we have
\beeq{\label{tildesigma-1}
\dot{\tilde\sigma} = \mathbf{H}L_{\ell}\tilde\eta + \dot\sigma
}
where $L_{\ell}=\diag(\ell_1 I_{r_1},\ldots,\ell_mI_{r_m})$, and
\beeq{\ba{l}\label{H}
\mathbf{H} = \mbox{blkdiag}(H_1,\ldots,H_m)\,,\\ H_k = \qmx{0 & \cdots & 0 & \gamma_{k,r_k}^2 & - \gamma_{k,r_k}^2 & 0}\in\mathbb{R}^{2r_k}.
\ea}

Recalling (\ref{sigma}) and (\ref{actual-u}), we observe that
\[\ba{rcl}
\sigma
= a(x) - \Delta_b(x)\mbox{satv}_{l}(\sigma + K\xi - \tilde\sigma \, - K(\Lambda_{\ell}^{-1}\otimes C_2)\tilde\eta)\,\\
\ea\]
whose derivative, by setting $\Delta_0 := \Delta_b(x)\nabla\mbox{satv}_{l}$, is given by
\[\ba{l}
\dot\sigma = \dot a(x) -\dot b(x)\widehat {\mathsf{B}}^{-1}\mbox{satv}_{l}(\sigma + K\xi - \tilde\sigma  - K(\Lambda_{\ell}^{-1}\otimes C_2)\tilde\eta)\,\\ - \Delta_0[\dot\sigma + K\dot\xi - \dot{\tilde\sigma} - K(\Lambda_{\ell}^{-1}\otimes C_2)\dot{\tilde\eta})]\,.\\
\ea\]
By adding  $\Delta_0\dot\sigma$ on both sides of the above equation and setting
\beeq{\ba{l}\label{Delta-1}
\Delta_1=\dot a(x) -\dot b(x)\widehat {\mathsf{B}}^{-1}\mbox{satv}_{l}(\phi(x,\tilde\eta) - \tilde\sigma \,\\ \qquad - K(\Lambda_{\ell}^{-1}\otimes C_2)\tilde\eta) - \Delta_0K\dot\xi\,,
\ea}
we have
\beeq{\label{dsigma-2}\ba{l}
(I_m+\Delta_0)\dot\sigma = \Delta_1 + \Delta_0(\dot{\tilde\sigma} + K(\Lambda_{\ell}^{-1}\otimes C_2)\dot{\tilde\eta})\,.\\
\ea}
We then observe that $\Delta_0$ and $\Delta_1$ have the following properties.
\begin{lemma}\label{lemma-Delta}
With Assumption \ref{ass-b}, for all $x\in\Omega_{c+1}$,
\begin{itemize}
  \item[(i)] $|\Delta_0|\leq\mu_0<1$, and $I_m+\Delta_0$ is invertible,
  \item[(ii)] there exists a constant $\delta_1>0$, independent of $\ell=\col(\ell_1,\ldots,\ell_k)$ such that $|\Delta_1|\leq\delta_1$ holds for all $\tilde\eta\in\mathbb{R}^{2\mathbf{r}}$.
\end{itemize}
\end{lemma}

The proof of Lemma \ref{lemma-Delta}.(i) is straightforward using Assumption \ref{ass-b} and the fact that $\nabla\mbox{satv}_{l}$ is a diagonal matrix whose entries are less than one, while  the proof of (ii) can be easily concluded by deriving the explicit expression of $\Delta_1$ and is also omitted.
Using the first part of Lemma \ref{lemma-Delta}, (\ref{dsigma-2}) implies
\beeq{\label{d-sigma}
\dot\sigma = (I_m+\Delta_0)^{-1}[\Delta_1 + \Delta_0(\dot{\tilde\sigma} + K(\Lambda_{\ell}^{-1}\otimes C_2)\dot{\tilde\eta})]\,.
}

It can be verified that $(\Lambda_{\ell}^{-1}\otimes C_2){\tilde\eta}$ is independent of $\tilde\eta_{k,r_k}^2$, $k=1,\ldots,m$, and thus $(\Lambda_{\ell}^{-1}\otimes C_2)\dot{\tilde\eta}$ can be expressed as a linear function of $\tilde\eta$ from (\ref{tilde-eta-1})-(\ref{tilde-eta-k}). Namely,
\beeq{\label{c-d-eta}
(\Lambda_{\ell}^{-1}\otimes C_2)\dot{\tilde\eta} = J(\ell)\tilde\eta
}
where $J(\ell)$ is a matrix dependent of $\ell$. To be precise, bearing in mind the definition of $\Lambda_{\ell}$ given after (\ref{u-error}),  $J(\ell)$ has the property that for any $\ell_i \geq 1$, $i=1,\ldots,m$, there exists $\delta_2>0$, independent of $\ell_i$'s, such that
\beeq{\label{ineq-J}
|J(\ell)|\leq \delta_2\,.
}

Substituting (\ref{d-sigma}) and (\ref{c-d-eta}) into (\ref{tildesigma-1}), we obtain
\[\ba{l}
[I_m-(I_m+\Delta_0)^{-1}\Delta_0]\dot{\tilde\sigma}\,\\ = \mathbf{H}L_{\ell}\tilde\eta + (I_m+\Delta_0)^{-1}[\Delta_1 + \Delta_0KJ(\ell)\tilde\eta)]\,.
\ea\]
By
$
[I_m-(I_m+\Delta_0)^{-1}\Delta_0] = (I_m+\Delta_0)^{-1}\,,
$
we further obtain
\beeq{\ba{l}
\dot{\tilde\sigma} = (I_m+\Delta_0)\mathbf{H}L_{\ell}\tilde\eta + \Delta_1 + \Delta_0KJ(\ell)\tilde\eta\,.
\ea}

Thus, the equations of the re-scaled estimation errors (\ref{tilde-eta-1}) and (\ref{tilde-eta-k}) can be compactly described by
\beeq{\label{tilde-eta}\ba{l}
\dot{\tilde\eta}= [\mathbf{F}+\mathbf{G}\Delta_0\mathbf{H}+\mathbf{G}\Delta_0KJ(\ell)L_{\ell}^{-1}]L_{\ell}\tilde\eta  + \mathbf{G} \Delta_1
\ea}
where
\[\ba{l}
\mathbf{F} = \qmx{ F_1 & 0 & \cdots & 0  \cr
                        \frac{1}{\ell_1}L_{21}(\ell_2,y)B_{2r_1}^{\top} &  F_2 & \cdots & 0  \cr
                        \vdots & \cdots & \ddots & \cdot  \cr
                        \frac{1}{\ell_1}L_{m1}(\ell_m,y)B_{2r_1}^{\top} & \frac{1}{\ell_2}L_{m2}(\ell_m,y)B_{2r_2}^{\top} & \cdots & \,F_m  \cr
                }\,\\
\mathbf{G} = \mbox{blkdiag}(B_{2r_1} ,\ldots, B_{2r_m})\,.
\ea\]
By Assumption \ref{ass-delta}, it is clear that given $\ell_i\geq 1$, there exists $\iota_{ij}>0$, independent of $\ell_i$ such that
\beeq{\label{eq:Lij}
|L_{ij}(\ell_i,y)|\leq \iota_{ij}\ell_i^{r_i-r_j+1}\,.
}

\subsection{Stability Analysis of the Estimation Error Dynamics (\ref{tilde-eta})}

Before presenting the main result of this subsection, two fundamental lemmas, proven in Appendix, are given below.
\begin{lemma}
\label{lemma-3}
Suppose Assumption \ref{ass-b} holds. There exist  symmetric positive definite matrices $P_i$ and positive constants $\lambda_i>0$, $i=1,\ldots,m$ such that
\[
\dst\sum_{i=1}^m \tilde\eta_i^{\top} (P_i F_i+F_i^{\top}P_i) \tilde\eta_i + 2\tilde\eta^{\top}P\mathbf{G}\Delta_0\mathbf{H}\tilde\eta \leq -\sum_{i=1}^m\lambda_i |\tilde\eta_i|^2
\]
 with $P=\blkdiag(P_1,\ldots,P_m)$, holds for all $x\in\Omega_{c+1}$.
\end{lemma}

Let $\lambda_{\min}=\min\{\lambda_1,\ldots,\lambda_m\}$, $\varrho_0={\lambda_{\min}}|P|^2\mu_0^2|K|^2\delta_2^2/8$ and $\varrho_1={8}\delta_1^2|P|^2/{\lambda_{\min}}$. We present the following lemma.
\begin{lemma}\label{lemma-5}
There exist constants $g_i>0$, $i=1,\ldots,m$, independent of $\kappa$, and $\theta^\ast>1$ such that for all $\kappa \geq\theta^\ast$,
\[\ba{l}
\frac{\lambda_m}{2}\ell_m^2-\varrho_0 \geq \kappa \ell_m\\
\frac{\lambda_i}{4}\ell_i^2-\sum_{j=i+1}^m \frac{2(j-1)|P_j|^2\iota_{ji}^2}{\lambda_j}\ell_j^{2(r_j-r_i+1)}-\varrho_0 \geq \kappa \ell_i
\ea\]
hold with
\beeq{\label{c-ell}\ba{l}
\ell_m=g_m\kappa \,\\
\ell_{i}= g_{i}\cdot(\ell_{i+1})^{r_{i+1}-r_{i}+1}\,,\quad \mbox{ for } 1\leq i \leq m-1\,.
\ea}

\end{lemma}

Bearing in mind the above two lemmas, the stability property of (\ref{tilde-eta}) is formulated as below.
\begin{proposition}\label{lemma-4}
Given any $\tau_{\max}>0$ and $R>0$, suppose $x\in\Omega_{c+1}$ for all $t\in[0,\tau_{\max})$, and the initial conditions $|\eta(0)|\leq R$.
Let $g_i$, $\ell_i$, $i=1,\ldots,m$ and $\theta^\ast$ be as in Lemma \ref{lemma-5}. Then for every $\tau_2<\tau_{\max}$ and every $\epsilon>0$, there exists a $\kappa^{\ast}\geq\theta^\ast$ such that for all $\kappa \geq\kappa^{\ast}$, $|\tilde\eta(t)|\leq 2\epsilon$ for all $t\in[\tau_2,\tau_{\max})$.
\end{proposition}
\begin{proof}
Let $V_c(\tilde\eta)=\tilde\eta^{\top}L_{\ell}P\tilde\eta$, and $\alpha_1=\min\{\mbox{eig}(P)\}$ and $\alpha_2=\max\{\mbox{eig}(P)\}$ with $\mbox{eig}(P)$ denoting the set of all eigenvalues of matrix $P$. It is clear that
\beeq{\ba{l}\label{ineq-V-c-0}
V_c(\tilde\eta)\geq \alpha_1\sum_{i=1}^m\ell_{i}|\tilde\eta_i|^2 \geq  \alpha_1\ell_{\min}|\tilde\eta|^2\,\\
V_c(\tilde\eta)\leq \alpha_2\sum_{i=1}^m\ell_{i}|\tilde\eta_i|^2 \leq  \alpha_2\ell_{\max}|\tilde\eta|^2\,
\ea}
where $\ell_{\max},\ell_{\min}$ denote the maximum and minimum of $\ell_1,\ldots,\ell_m$, respectively.

We compute the derivative of $V_c$ along system (\ref{tilde-eta}) as
\[\ba{l}
\dot V_c
= \dst\sum_{i=1}^m \ell_i\tilde\eta_i^{\top} (P_i F_i+F_i^{\top}P_i) \ell_i\tilde\eta_i + 2\tilde\eta^{\top}L_{\ell}P\mathbf{G}\Delta_0\mathbf{H}L_{\ell}\tilde\eta \\ \qquad + \dst 2\sum_{i=1}^{m-1}\sum_{j=i+1}^{m}\ell_j\tilde\eta_j^{\top}P_jL_{ji}(\ell_j,y)B_{2r_i}^{\top}\tilde\eta_i \,\\
\qquad + 2\tilde\eta^{\top}L_{\ell}P\mathbf{G}[\Delta_0KJ(\ell)\tilde\eta  + \Delta_1]\,\\
\leq -\dst\sum_{i=1}^m\lambda_i\ell_i^2|\tilde\eta_i|^2 + \dst 2\sum_{i=1}^{m-1}\sum_{j=i+1}^{m}\iota_{ij}|P_j|\ell_i^{r_i-r_j+2}|\tilde\eta_i|\cdot|\tilde\eta_j|\\
 \qquad + 2\tilde\eta^{\top}L_{\ell}P\mathbf{G}[\Delta_0KJ(\ell)\tilde\eta  + \Delta_1]
\ea\]
where the inequality is obtained by  Lemma \ref{lemma-3} and (\ref{eq:Lij}).
Then by Young's Inequality {\cite[pp. 466]{Khalil(book2002)}}, (\ref{ineq-J}) and Lemma \ref{lemma-Delta}, we have
\[\ba{l}
\dst2\iota_{ij}|P_j|\ell_i^{r_i-r_j+2}|\tilde\eta_i|\cdot|\tilde\eta_j| \leq \\ \qquad \dst\frac{2(j-1)\iota_{ji}^2|P_j|^2}{\lambda_j}\ell_j^{2(r_j-r_i+1)}|\tilde\eta_i|^2+\frac{\lambda_j\ell_j^2}{2(j-1)}|\tilde\eta_j|^2\,,\\
\dst2\tilde\eta^{\top}L_{\ell} P\mathbf{G}\Delta_0KJ(\ell)\tilde\eta \leq \\ \qquad \dst\frac{\lambda_{\min}}{8}\sum_{k=1}^m\ell_k^2|\tilde\eta_k|^2 + \dst\frac{8}{\lambda_{\min}}|P|^2\mu_0^2|K|^2\delta_2^2\sum_{k=1}^m|\tilde\eta_k|^2\,,\\
\dst2\tilde\eta^{\top}L_{\ell}\tilde\eta P\mathbf{G}\Delta_1 \leq \dst\frac{\lambda_{\min}}{8}\sum_{k=1}^m\ell_k^2|\tilde\eta_k|^2 + \dst\frac{8}{\lambda_{\min}}|P|^2(\delta_1)^2\,.
\ea\]
The first of the above inequalities further indicates that
\[\ba{l}
\dst2\sum_{i=1}^{m-1}\sum_{j=i+1}^{m}\iota_{ij}|P_j|\ell_i^{r_i-r_j+2}|\tilde\eta_i|\cdot|\tilde\eta_j| \leq \frac{\lambda_m}{2}\ell_m^2|\tilde\eta_m|^2+ \\ \dst\sum_{i=1}^{m-1}\big(\frac{\lambda_i}{2}\ell_i^2+\dst\sum_{j=i+1}^m \frac{2(j-1)\iota_{ji}^2|P_j|^2}{\lambda_j}\ell_j^{2(r_j-r_i+1)}\big)|\tilde\eta_i|^2\,.
\ea\]
Therefore, we  have
\[\ba{l}
\dot V_c \leq  -(\dst\frac{\lambda_m}{2}\ell_m^2-\varrho_0)|\tilde\eta_m|^2+ \varrho_1\,\\
-\dst\sum_{i=1}^{m-1}\big(\frac{\lambda_i}{4}\ell_i^2-\dst\sum_{j=i+1}^m \frac{2(j-1)\iota_{ji}^2|P_j|^2}{\lambda_j}\ell_j^{2(r_j-r_i+1)}-\varrho_0\big)|\tilde\eta_i|^2.
\ea\]


By Lemma \ref{lemma-5}, it is seen that for $\kappa \geq\theta^\ast$, the derivative of $V_c$  can be further bounded by
\vspace{-0.5em}
\[
\dot V_c \leq -\kappa \sum_{i=1}^{m}\ell_i|\tilde\eta_i|^2 + \varrho_1\,
\leq-({\kappa }/{\alpha_2}) V_c(\tilde\eta) + \varrho_1\,.
\vspace{-0.5em}
\]

\noindent
\vspace{-0.5em}
Recalling inequalities (\ref{ineq-V-c-0}), standard arguments then show that
\[\ba{l}
\qquad V_c(\tilde\eta(t)) \leq e^{-\frac{\kappa }{\alpha_2}t}V_c(\tilde\eta(0)) + {\varrho_1\alpha_2}/{\kappa }\,\Longrightarrow  \\
|\tilde\eta(t)| \leq {\alpha_2\ell_{\max}}/{(\alpha_1\ell_{\min})}e^{-\frac{\kappa }{\alpha_2}t}|\tilde\eta(0)| + {\varrho_1\alpha_2}/({\alpha_1\kappa \ell_{\min}})\,.
\ea\]

With (\ref{c-ell}), and recalling the fact that all coefficients $g_i$'s in (\ref{c-ell}) are independent of parameter $\kappa$, it is immediate that there exists a $\tilde\kappa^\ast>0$ such that for all $\kappa\geq\max\{\tilde\kappa^\ast,\theta^\ast\}$, $\ell_{\min}=\ell_m$ and $\ell_{\max}=\ell_1$, and there exists a constant $\varsigma_1>0$, independent of $\kappa $ such that
$
({\alpha_2\ell_{\max}})/({\alpha_1\ell_{\min}}) \leq \varsigma_1 \kappa ^{\varpi_1}\,
$
where $\varpi_1=\dst\Pi_{k=1}^{m-1}(r_{k+1}-r_{k}+1)-1$.
On the other hand, since $|\eta(0)|\leq R$ and $x\in\Omega_{c+1}$, it can be seen from (\ref{tilde-eta}) that
$
|\tilde\eta(0)| \leq \varsigma_2\kappa ^{\varpi_2}
$
for some $\varsigma_2>0$ and $\varpi_2>0$.
Thus, we have
\beeq{
|\tilde\eta(t)| \leq \varsigma_1\varsigma_2\kappa ^{\varpi_1+\varpi_2}e^{-\frac{\kappa }{\alpha_2}t} + {\varrho_1\alpha_2}/({g_m\alpha_1\kappa^2})\,.
}
Fix any $\epsilon>0$. We know that for any $\tau_2\in(0,\tau_{\max})$ there always exists a $\bar\kappa^\ast>0$ such that for all $\kappa\geq\bar\kappa^\ast$ and $t\in[\tau_2,\tau_{\max})$,
$
\varsigma_1\varsigma_2\kappa ^{\varpi_1+\varpi_2}e^{-\frac{\kappa }{\alpha_2}t} \leq \epsilon\,.
$
Thus letting $\kappa^\ast=\max\{1,\theta^\ast,\tilde\kappa^\ast,\bar\kappa^\ast,\sqrt{{\varrho_1\alpha_2}/({g_m\alpha_1\epsilon})}\}$  completes the proof. $\square$

\end{proof}


\vspace{-1.5em}
\subsection{Stability Analysis of the Closed-Loop System}

In this subsection, we analyze the asymptotic stability of the resulting closed-loop system using the nonlinear separation principles \cite{Teel(1995),Isidori(1999)}.

\begin{theorem}
Consider the closed-loop system consisting of the plant (\ref{partialform}), the observers (\ref{ELPHGO-1})-(\ref{ELPHGO-2}),  and the controller {(\ref{eq:estimates})-}(\ref{actual-u}). Suppose Assumptions \ref{ass-MPP}-\ref{ass-b} are satisfied. Given any compact set $\mathcal{C}\subset\mathbb{R}^{n+2\mathbf{r}}$, there exist $\ell_i>1$, $i=1,\ldots,m$,  such that all trajectories of the closed-loop system with initial conditions $(x(0),\eta(0))\in\mathcal{C}$ remain bounded and  $\dst\lim_{t\rightarrow\infty}|x(t)|=0$.
\end{theorem}
\begin{proof}
Observe that (\ref{partialform}) can be rewritten as
\beeq{
\ba{l}
\dot x_0 = f_0(x_0,\xi,u)\,\\
\dot \xi = [\mathbf{A} - \mathbf{B}(y)K]\xi + \mathbf{B}(y)[K\xi+ a(x)+b(x)u]\\
\ea
}
which, together with (\ref{V_x}), yields
\[\ba{l}
\dst\dot V_x(x) \leq - \alpha_x(|x|) + \dst\frac{\partial V_x}{\partial \xi}\mathbf{B}(y)[K\xi+ a(x)+b(x)u]\\
\dst\leq \left|\frac{\partial V_x}{\partial \xi}\mathbf{B}(y)\right| \left[|K\xi|+|a(x)|+ (l+\epsilon_0)|b(x)||\widehat {\mathsf{B}}^{-1}|\right]\,.
\ea\]
It is clear that there exists a $\delta_0>0$, independent of $\ell_k$, such that
$\dot V_x(x) \leq \delta_0$
holds for all $x\in \Omega_{c+1}$.

Therefore, it can be concluded that given any initial condition $x(0)\in \mathcal{C}_x\subset\Omega_{c}$, there exists $\tau_1\geq{1\over\delta_0}$ such that $x(t)\in\Omega_{c+1}$ for all $t\in[0,\tau_1]$.

Now, let us consider the closed-loop system (\ref{xi-3})-(\ref{tilde-eta}) as
\beeq{\label{clc-loop}
\ba{l}
\dot x_0 = f_0(x_0,\xi,u)\,\\
\dot \xi = (\mathbf{A} - \mathbf{B}(y)K)\xi + \mathbf{B}(y)[\phi(x,\tilde\eta)\\ \qquad -\mbox{satv}_l\left(\phi(x,\tilde\eta) - \tilde\sigma - K(\Lambda_{\ell}^{-1}\otimes C_2)\tilde\eta\right)]\\
\dot{\tilde\eta} = [\mathbf{F}+\mathbf{G}\Delta_0\mathbf{H}+\mathbf{G}\Delta_0KJ(\ell)L_{\ell}^{-1}]L_{\ell}\tilde\eta  + \mathbf{G} \Delta_1\,.
\ea}

{
According to the Mean Value Theorem, we know from (\ref{u-error}) that there exists $u^\prime\in\mathbb{R}^m$ such that
\[\ba{l}
\widehat {\mathsf{B}}\pi(x,\tilde\eta)=\,\\ - \nabla\mbox{satv}_{l}(u^\prime)(\sigma(x,\pi(x,\tilde\eta)) + K\xi - \tilde\sigma - K(\Lambda_{\ell}^{-1}\otimes C_2)\tilde\eta)\,,
\ea\]
which, for $x\in\Omega_{c+1}$, yields
\[\ba{l}
\pi(x,\tilde\eta) = -\widehat {\mathsf{B}}^{-1}[I+\nabla\mbox{satv}_{l}(u^\prime)\Delta_b(x)]^{-1}\nabla\mbox{satv}_{l}(u^\prime)\,\\ \qquad \cdot[a(x)+K\xi - \tilde\sigma - K(\Lambda_{\ell}^{-1}\otimes C_2)\tilde\eta]\,.
\ea\]
Recalling the definition of $\phi(x,\tilde\eta)$ in (\ref{eq-phi}), we thus have
\[\ba{l}
\phi(x,\tilde\eta)
=(I-\Delta_3)[K\xi+a(x)] + \Delta_3[\tilde\sigma +K(\Lambda_{\ell}^{-1}\otimes C_2)\tilde\eta]\,\\
\ea\]
with $\Delta_3: = \Delta_b(x)[I+\nabla\mbox{satv}_{l}(u^\prime)\Delta_b(x)]^{-1}\nabla\mbox{satv}_{l}(u^\prime)$. For all $x\in\Omega_{c+1}$ , by $|\nabla\mbox{satv}_{l}(u^\prime)|\leq 1$ and (\ref{bproperty}), we have
\[
|\Delta_3|
\leq |\Delta_b(x)|\cdot\left|[I+\nabla\mbox{satv}_{l}(u^\prime)\Delta_b(x)]^{-1}\right| \leq {\mu_0}/({1-\mu_0})\,,
\]
which yields
\[\ba{l}
|\phi(x,\tilde\eta) - \tilde\sigma - K(\Lambda_{\ell}^{-1}\otimes C_2)\tilde\eta|\,\\
\leq \frac{1}{1-\mu_0}|K\xi+a(x)| + \frac{1}{1-\mu_0}|\tilde\sigma +K(\Lambda_{\ell}^{-1}\otimes C_2)\tilde\eta|\,\\
\leq l -1 + \rho|\tilde\eta|
\ea\]
where the last inequality is obtained by setting $\rho:=({|K|+1})/({1-\mu_0})$, using  the definition of the saturation level $l$ in (\ref{eq:l}) and $\ell_i\geq 1$, $i=1,\ldots,m$.
}

In view of the above analysis, given any $\tau_2<\tau_1$, according to Proposition \ref{lemma-4},  for any sufficiently small $\epsilon>0$ there exists a sufficiently large $\kappa$ such that $\rho|\tilde\eta(t)|\leq 2\rho\epsilon<1$  for all $t\in(\tau_2,\tau_1]$.
This implies that
$$\ba{l}
\mbox{satv}_l(\phi(x,\tilde\eta) - \tilde\sigma - K(\Lambda_{\ell}^{-1}\otimes C_2)\tilde\eta)  \\ \qquad \qquad =\phi(x,\tilde\eta) - \tilde\sigma - K(\Lambda_{\ell}^{-1}\otimes C_2)\tilde\eta
\ea$$
for $t\in(\tau_2,\tau_1]$. Thus, the $\xi$-subsystem in (\ref{clc-loop}) reduces to
\[
\dot \xi = (\mathbf{A} - \mathbf{B}(y)K)\xi + \mathbf{B}(y) [\tilde\sigma + K(\Lambda_{\ell}^{-1}\otimes C_2)\tilde\eta]\,
\]
where $|\tilde\sigma + K(\Lambda_{\ell}^{-1}\otimes C_2)\tilde\eta| \leq 2\rho \epsilon$ for $t\in(\tau_2,\tau_1]$.

Towards this end, pick any number $0<c'\ll c$ and consider the ``annular" compact set
$
S_{c'}^{c+1} = \{x: c'\leq V_x(x)\leq c+1\}\,.
$
Let $\nu_{\min}$ be $\nu_{\min} = \min_{x\in S_{c'}^{c+1}}\alpha_x(|x|)$, and $\epsilon$ be such that $2\rho \epsilon \sup\limits_{x\in\Omega_{c+1}}|\frac{\partial V_x}{\partial \xi}\mathbf{B}(y)|\leq {1\over 2}\nu_{\min}$. It then follows that $
\dot V_x(x) \leq -{1\over 2}\nu_{\min}
$
so long as $x\in S_{c'}^{c+1}$. This, in turn, implies
\[\ba{l}
V_x(x(t)) \leq V_x(x(\tau_1)) - {\nu_{\min}\over 2}(t-\tau_1) \leq c+1- {\nu_{\min}\over 2}(t-\tau_1)
\ea\]
so long as $x\in S_{c'}^{c+1}$. Clearly, there exists a time $\tau_3>\tau_1$ such that $x(t)\in\Omega_{c+1}$ for all $t\in[\tau_1,\tau_3]$ and $V_x(x(\tau_3))=c'$. Since $\dot V_x$ is negative on the boundary of $\Omega_{c'}$, it is concluded that $x(t)\in\Omega_{c'}$ for all $t\geq\tau_3$ and $x(t)\in\Omega_{c+1}$ for all $t\geq0$.

Therefore, according to the standard arguments in \cite{Teel(1995),Isidori(1999)}, and by the local exponential part of Assumption \ref{ass-MPP}, we can conclude that the origin of (\ref{clc-loop}) is asymptotically stable for all  $(x(0),\eta(0))\in\mathcal{C}$. This thus completes the proof. $\square$

\end{proof}

{
\begin{remark}
  As in \cite{Teel(1995),Isidori(1999)}, if Assumption \ref{ass-MPP} is relaxed by removing the locally exponential element, then the resulting closed-loop system would be semiglobally practically stable at the origin. We also remark that the observers (\ref{ELPHGO-1})-(\ref{ELPHGO-2}) belong to the class of extended-state observers \cite{Freidovich&Khalil(TAC2008),WangSCL,Wu&Isidori(2019),Guo&Zhao(2013)}, providing extra estimates of perturbation terms. In this respect, our approach is also robust to bounded smooth perturbations in functions $f_0,a_k,b_k$ in (\ref{partialform}) as in \cite{Freidovich&Khalil(TAC2008)}, leading to semiglobally practical stability at the origin. $\square$
\end{remark}
}

\section{An Illustrative Example}
Consider the 2-input 2-output system of the form
\beeq{\label{eq:exp-sys}
\ba{l}
\dot x_{01} = -x_{01}+x_{02}\xi_{12}u_2 + \xi_{12}\,\\
\dot x_{02} = -x_{02}-x_{01}\xi_{12}u_2 + \xi_{11}\,\\
\left\{
\ba{l}
\dot \xi_{1} = A_2\xi_1 + B_2[a_1(x) + b_1(x)u]\\
y_1 = \xi_{11}
\ea
 \right.\\
\left\{\ba{l}
\dot \xi_{2} = A_3\xi_{2} + M_2^1(y)[a_1(x) + b_1(x)u]\\ \qquad  + B_3[ a_2(x) + b_2(x)u]\,\\
y_2 = \xi_{21}
\ea\right.
\ea}
where the full state $x=\col(x_0,\xi_1,\xi_2)\in\mathbb{R}^7$ with $x_0=\col(x_{01},x_{02})$, $\xi_{1}=\col(\xi_{11},\xi_{12})$, $\xi_{2}=\col(\xi_{21},\xi_{22},\xi_{23})$, and output $y=\col(y_1,y_2)\in\mathbb{R}^2$ and $u=\col(u_1,u_2)\in\mathbb{R}^2$, $a(x)=\qmx{x_{01}\xi_{21} \cr x_{02}}$, $b(x)=\qmx{1 & \frac{\sin\xi_{12}}{3} \cr 0&1}$ and the multiplier vector $M_2^1(y)=\col(0,\delta_{23}^1(y),0)$ with $\delta_{23}^1(y)=\cos y_1$.  It is clear that Assumption 1 is satisfied with an ISS Lyapunov function $V_0(x_0)=\frac{1}{2}|x_{0}|^2$, whose derivative is given by
\[
\dot V_0 = -|x_{01}|^2-|x_{02}|^2 + x_{01}\xi_{12} + x_{02}\xi_{11} \leq -\frac{1}{2}|x_0|^2 + \frac{1}{2}|\xi_{1}|^2\,.
\]
Besides, it can be easily verified that Assumption 2 is satisfied, i.e., $|\delta_{22}^1(y)| \leq 1$ for all $y\in\mathbb{R}^2$, and Assumption 3 is also satisfied with $\widehat{{\mathsf{B}}} = I_2$ and $\mu_0=\frac{1}{3}$ for $x\in\mathbb{R}^7$, i.e., $|(b(x)-\widehat{{\mathsf{B}}})\widehat{{\mathsf{B}}}^{-1}|\leq \mu_0$.

\begin{remark}
  We observe that if the multiplier $\delta_{23}^1$ is zero, then the system (\ref{eq:exp-sys}) takes a standard MIMO normal form with the strongly minimum-phase property, for which the approaches in \cite{Freidovich&Khalil(TAC2008),WangSCL,Guo&Zhao(2013)} can be applied. If $\delta_{23}^1$ is a nonzero constant, then the system (\ref{eq:exp-sys}) turns out linearizable by static state feedback. For this class of systems, the approach in \cite{WangTAC} can be adopted by firstly designing ``virtual" outputs $\bar y_1 =  c_{11} \xi_{11} + \xi_{12}$ and $\bar y_2 = c_{21}\xi_{21} + c_{22}(\xi_{22}-\delta_{23}^1\xi_{12}) + \xi_{23}$ so as to render a strongly minimum-phase MIMO system in the normal form, for which the conventional output feedback stabilization approaches can be adopted. We note that both cases are indeed special cases of our considered systems, while the approaches in \cite{Freidovich&Khalil(TAC2008),WangSCL,Guo&Zhao(2013),WangTAC} cannot be used to handle our problem since the strongly minimum-phase property is not preserved after redesigning the ``virtual" outputs for normal form due to the presence of the output-dependent multiplier $\delta_{23}^1(y)$. If $\delta_{23}^1(\cdot)$ is also dependent on other states, e.g., $\delta_{23}^1(y,\xi_{12},\xi_{22})$, and there is no dynamics of the internal state $x_0$, i.e., with a trivial zero dynamics, the system (\ref{eq:exp-sys}) belongs to the class of systems considered in \cite{Wang&Isidori(TAC2017),Wu&Isidori(2019)}, where the system is linearizable via dynamical state feedback. However, we note that when the zero dynamics is nontrivial as the $x_0$ subsystem in (\ref{eq:exp-sys}), how to apply the solutions of \cite{Wang&Isidori(TAC2017),Wu&Isidori(2019)} is not clear. Moreover, \cite{Wu&Isidori(2019)} requires the high-frequency gain matrix $b(x)$ to be lower-triangular, while our $b(x)$ in (\ref{eq:exp-sys}) does not satisfy it. In view of the above analysis, the existing approaches cannot be used to handle the output feedback stabilization of the system (\ref{eq:exp-sys}). $\square$
\end{remark}

In this setting, the ideal feedback control can be designed as
\beeq{\label{eq:u^ast}\ba{l}
u^\ast = -b(x)^{-1}[a_1(x) + K \xi]\,\\
 = -\qmx{1 & -\frac{\sin\xi_{12}}{3} \cr 0&1}\qmx{x_{01}\xi_{21} +K_1\xi_1 \cr x_{02}+K_2\xi_2}
\ea}
with $K_1=\qmx{\frac{1}{4} & 1}$ and $K_2=\qmx{\frac{1}{8} & \frac{3}{4}&\frac{3}{2}}$.
This in turn yields a $\xi$-dynamics of the form
\beeq{\label{eq:exp-xi}
\ba{l}
\left\{
\ba{l}
\dot \xi_{11} = \xi_{12}\\
\dot \xi_{12} = -K_1\xi_1\\
\ea
 \right.\\
\left\{\ba{l}
\dot \xi_{21} = \xi_{22}\\
\dot \xi_{22} = \xi_{23} - \delta_{23}^1(y)K_1\xi_1\\
\dot \xi_{23} = -K_2\xi_2\,\\
\ea\right.
\ea}
which is clearly globally exponentially stable at the origin with a Lyapunov function
$V_{\xi}(\xi) =  \xi_1^\top P_{\xi_1}\xi_1 + \xi_2^\top P_{\xi_2}\xi_2$ with
\[\ba{l}
P_{\xi_1} = 100\qmx{2.625 &2 \cr 2 & 2.5}\,,\\
P_{\xi_2} = \qmx{4.2266 & 6.8594 & 4.0000 \cr 6.8594&15.8750&9.8125 \cr 4.0000&9.8125&6.8750}\,.
\ea\]
Computing the time derivative of $V_{\xi}(\xi)$ yields
\[
\dot V_{\xi}  \leq  -100|\xi_1|^2 - |\xi_2|^2 + 16|\xi_1||\xi_2| \leq -2|\xi_1|^2 - 0.3|\xi_2|^2\,.
\]
Towards this end, the ideal feedback control law (\ref{eq:u^ast}) globally exponentially stabilizes system (\ref{eq:exp-sys}) with a Lyapunov function $V_x(x)=V_0(x_0) + V_{\xi}(\xi)$, satisfying
\[
\dot V_x(x) \leq -\frac{1}{2}|x_{0}|^2-|\xi_1|^2 - 0.3|\xi_2|^2\,.
\]

With the above ideal control law (\ref{eq:u^ast}) and the Lyapunov function $V_x(x)$, we take $c=2$ and $\Omega_c=\{x\in\mathbb{R}^7: V_x(x)\leq c\}$. Then, it can be verified that $ \sup_{x\in\Omega_{c+1}}\frac{1}{1-\mu_0}\left|a(x)+K\xi\right|\leq 24$, and we thus design the saturation level $l=25$. The desired output-feedback control law is given by
\beeq{\label{actual-u}
u = -\widehat B^{-1}\mbox{satv}_{l}\left(\hat\sigma + K\hat\xi\right)
}
where $\hat\sigma= \col(\eta_{12}^2,\eta_{23}^2)$ and $\hat\xi=\col(\eta_{11}^1,\eta_{12}^1,\eta_{21}^1,\eta_{22}^1,\eta_{23}^1)$ are provided by the following extended low-power high-gain observer
\beeq{\label{ELPHGO-exp}\ba{l}
\left\{
\ba{l}
\dot \eta_{11}^1 = \eta_{11}^2 +  \ell_1\gamma_{11}^1(y_1-\eta_{11}^1)\,\\
\dot \eta_{11}^2 = \eta_{12}^2 + \widehat B_1 u + (\ell_1)^2\gamma_{11}^2(y_1-\eta_{11}^1)\,
\ea
\right.\\
\left\{
\ba{l}
\dot \eta_{12}^1 = \eta_{12}^2 + \widehat B_1 u + \ell_1\gamma_{12}^1(\eta_{11}^2-\eta_{12}^1)\,\\
\dot \eta_{12}^2 =    (\ell_1)^2\gamma_{12}^2(\eta_{11}^2-\eta_{12}^1)\,\\
\ea
\right.\\
\left\{
\ba{l}
\dot \eta_{21}^1 = \eta_{21}^2 +  \ell_2\gamma_{21}^1(y_2-\eta_{21}^2)\,\\
\dot \eta_{21}^2 = \eta_{22}^2 + \delta_{23}^1(\eta_{12}^2+ \widehat B_2 u)+  (\ell_2)^2\gamma_{21}^2(y_2-\eta_{21}^2)\,
\ea
\right.\\
\left\{
\ba{l}
\dot \eta_{22}^1 = \eta_{22}^2 + \delta_{23}^1(\eta_{12}^2+ \widehat B_2 u) + \ell_2\gamma_{22}^1(\eta_{21}^2-\eta_{22}^1)\,,\\
\dot \eta_{22}^2 = \eta_{23}^2  + \widehat B_2 u  + (\ell_2)^2\gamma_{22}^2(\eta_{21}^2-\eta_{22}^1)\,
\ea
\right.\\
\bigg\{
\ba{l}
\dot \eta_{23}^1 = \eta_{23}^2 + \widehat B_2 u  + \ell_2\gamma_{23}^1(\eta_{22}^2-\eta_{23}^1)\,\\
\dot \eta_{23}^2 =  (\ell_2)^2\gamma_{23}^2(\eta_{22}^2-\eta_{23}^1)\,\\
\ea
\\
\ea}
in which $\qmx{\gamma_{i1}^1\cr \gamma_{i1}^2} = \qmx{2.5 \cr 4.6}$, $\qmx{\gamma_{i2}^1\cr \gamma_{i2}^2} = \qmx{2.5 \cr 1.533}$ for $i=1,2$, and $\qmx{\gamma_{23}^1\cr \gamma_{23}^2} = \qmx{2.5 \cr 0.511}$ such that
matrices
\[F_1 = \qmx{F_{11} & D_2  \cr
           \Gamma_{12}B_2^{\top} & F_{12}  \cr
            }\,,\quad
F_2 = \qmx{F_{21} & D_2 &  0  \cr
           \Gamma_{22}B_2^{\top} & F_{22} & D_2   \cr
            0 &  \Gamma_{23}B_2^{\top} & F_{23} \cr
            }
\]
are Hurwitz,
with  $F_{ki}=A_2-\Gamma_{ki}C_2$, and
\[\ba{l}
D_2 = \qmx{0 & 0 \cr 0 & 1}\,,\quad \Gamma_{ki}=\qmx{\gamma_{ki}^1 \cr \gamma_{ki}^2}\,, \quad i=1,\ldots,r_k\,,  k=1,2.
\ea\]

\begin{remark}
  The observers (\ref{ELPHGO-exp}) take advantage of the low-power high-gain technique proposed in \cite{Daniele&Lorenzo}. As seen from (\ref{ELPHGO-exp}), for each high-gain parameter $\ell_i$, $i=1,2$, it appears in (\ref{ELPHGO-exp}) with the power being either 1 or 2, independent on dimensions of partial states $\xi_1,\xi_2$.
  We note that in the standard extended high-gain observer \cite{Freidovich&Khalil(TAC2008),WangSCL}, the largest power of high-gain parameter is $r+1$ with $r$ being the dimension of states to be estimated. When such technique is adapted to handle the example (\ref{eq:exp-sys}), the observer takes the form
  \[\ba{l}
  \dot {\hat\xi}_1 = A_2\hat\xi_1 + B_2(\hat\sigma_1+ \widehat {\mathsf{B}}_1 u) + \diag(\ell_1,(\ell_1)^2)\bar\Gamma_1(y_1-C_2\hat\xi_1)\,\\
  \dot {\hat\sigma}_1 = (\ell_1)^3 \gamma_{13}(y_1-C_2\hat\xi_1)\,\\
  \dot {\hat\xi}_2 = A_3\hat\xi_2 + M_2^1(y)(\hat\sigma_1+ \widehat B_1 u) + B_3(\hat\sigma_2+ \widehat {\mathsf{B}}_2 u)\\ \qquad  + \diag(\ell_2,(\ell_2)^2,(\ell_2)^3)\bar\Gamma_2(y_2-C_3\hat\xi_2)\,\\
  \dot {\hat\sigma}_2 = (\ell_2)^4 \gamma_{24}(y_2-C_3\hat\xi_2)\,\\
  \ea
  \]
  with high-gain parameters $\ell_1,\ell_2$. It is clear that though the dimension of the observer reduces to $7$, the largest powers of high-gain parameters $\ell_1$ and $\ell_2$ increase to 3 and 4, respectively, which may cause numerical implementation problem when the gains $\ell_1,\ell_2$ are chosen very large. $\square$
\end{remark}

The simulations are finally performed  with high-gain parameters $\ell_1=5\times 10^5$ and $\ell_2=200$. The resulting evolutions of $|x(t)|$ are given in Figure \ref{fig1}, from which it can be seen that $|x(t)|$ asymptotically converges to zero and is lower than $5\times 10^{-5}$ at about $t=30s$. We are also interested in the robustness of the proposed output-feedback stabilizer. As a comparison, we add a sinusoidal perturbation $\sin t$ to the function $a_2(x)$, and the resulting evolutions of $|x(t)|$ are given in Figure \ref{fig2}, from which $|x(t)|$ is lower bounded by about $0.025$ at steady state.

\begin{figure}[thpb]
\begin{center}
\centering\includegraphics[height=60mm,width=75mm]{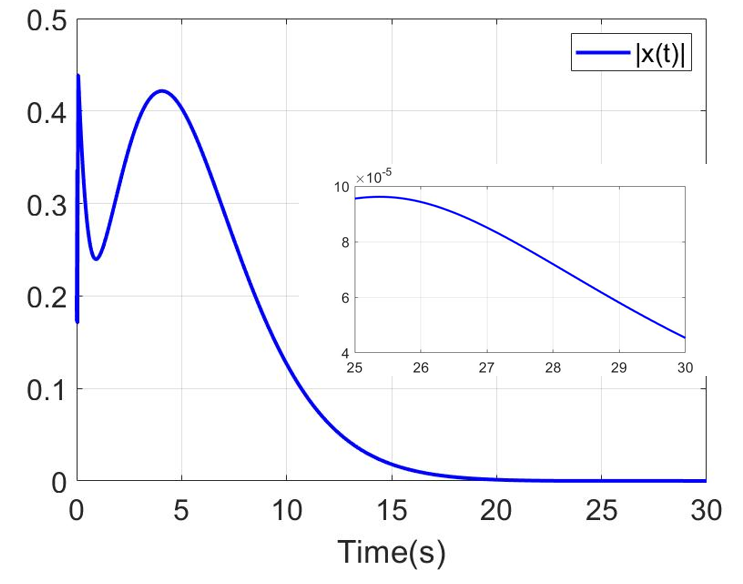} \caption{Evolution of $|x(t)|$.}
\label{fig1}
\end{center}
\end{figure}

\begin{figure}[thpb]
\begin{center}
\centering\includegraphics[height=60mm,width=75mm]{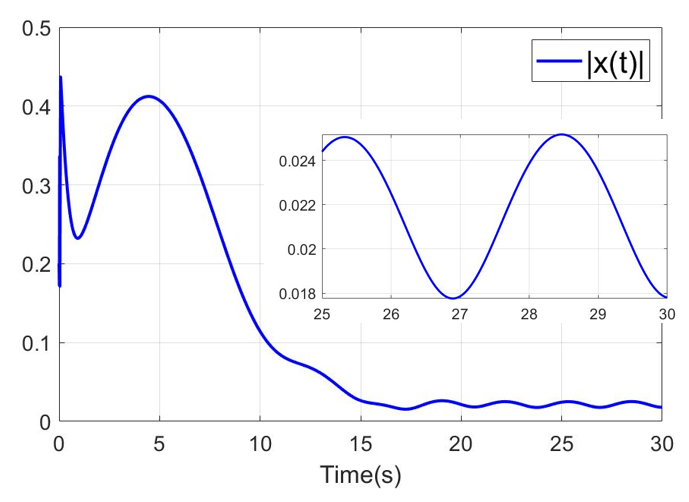} \caption{Evolution of $|x(t)|$ with perturbation.}
\label{fig2}
\end{center}
\end{figure}

\section{Conclusions}
This note studied the robust output feedback stabilization problem of multivariable invertible nonlinear systems (\ref{partialform}) with output-dependent multipliers. We first assumed that all states were accessible and proposed an ``ideal" state feedback law. By systematically designing a set of extended low-power high-gain observers, we showed that this ``ideal" law can be approximately estimated, providing a robust output feedback stabilizer such that the origin of the resulting closed-loop system is semiglobally asymptotically stable.
Moreover, each EHGO has the power of its high-gain parameter up to 2, which to some extent solves the numerical implementation problem.

\vspace{-0.5em}
\appendix

\subsection{Proof of Lemma \ref{lemma-3}}
\label{app-sec-3}

It is noted that $\gamma_{i,j}^1>0$ and  $\gamma_{i,j}^2>0$ are selected  such that matrix $F_i$ is Hurwitz \cite{Wang&Marconi(AUT2017)}. With these choices of $(\gamma_{i,j}^1,\gamma_{i,j}^2)$, we then consider the system
\beeq{\label{app-sys}\ba{rcl}
\dot {\tilde\eta}_i &=& F_i {\tilde\eta}_i + G_i u_i\,,\quad i=1,\ldots,m\\
u &=& \Delta_0y\,,\quad y_i = H_i {\tilde\eta}_i\,,\quad i=1,\ldots,m\\
\ea}
where state ${\tilde\eta}=\mbox({\tilde\eta}_1,\ldots,{\tilde\eta}_m)$ with ${\tilde\eta}_i=\col({\tilde\eta}_{i,1},\ldots,{\tilde\eta}_{i,r_i})$ and ${\tilde\eta}_{i,j}=\col({\tilde\eta}_{i,j}^1,{\tilde\eta}_{i,j}^2)$, output $y:=\col(y_1,\ldots,y_m)$ and input $u:=\col(u_1,\ldots,u_m)$.
By taking the change of variables
\[\ba{l}
\dst\chi_{i,k} = (\Pi_{j=1}^k\gamma_{i,r_i+1-j}^2)({\tilde\eta}_{i,r_i-k}^2-{\tilde\eta}_{i,r_i-k+1}^1)\,,\\ \qquad k=1,\ldots,r_i-1\\
\dst\chi_{i,r_i} = -(\Pi_{j=1}^{r_i}\gamma_{i,r_i+1-j}^2){\tilde\eta}_{i,1}^1\,\\
\dst\chi_{i,r_i+k} = -(\Pi_{j=1}^{r_i}\gamma_{i,j}^2){\tilde\eta}_{i,k}^2\,, k=1,\ldots,r_i \\
\ea\]
in which we have defined $y_i=\chi_{i,1}$,
system (\ref{app-sys}) is transformed into the lower-triangular form
\beeq{\label{app-sys-2}\ba{l}
\dot\chi_{i,k} = - \gamma_{i,r_i+1-k}^1\chi_k + \chi_{k+1} \,,\quad 1\leq k\leq r_i \\
\dot\chi_{i,r_i+k} = - (\Pi_{j=1}^k\gamma_{i,r_i+1-j}^2)\chi_{r_i+1-k} + \chi_{r_i+k+1}\,,\\ \qquad 1\leq k\leq r_i-1\\
\dot\chi_{i,2r_i} = - (\Pi_{j=1}^{r_i}\gamma_{i,j}^2)\chi_{1} -(\Pi_{j=1}^{r_i}\gamma_{i,j}^2)\,u_i \,\\
 u = \Delta_0y\,, \quad y_i = \chi_{i,1}\,\\
\ea}
Since $\gamma_{i,j}^1$ and  $\gamma_{i,j}^2$ are nonzero constants, the above change of variables defines a nonsingular matrix $T_i\in\mathbb{R}^{2r_i\times2r_i}$ such that $\chi_i=T_i \tilde\eta_i$ with $\chi_i=\col(\chi_{i,1},\ldots,\chi_{i,2r_i})$.

For compactness, system (\ref{app-sys-2}) can be rewritten as
\beeq{\label{app-sys-3}\ba{l}
\dot\chi_i = \bar F_i \chi_i + \bar G_{i} u_i\,,\quad
u = \Delta_0y\,,\quad y_i = \bar H_{i}\chi\,\quad \\
\ea}
in which $\bar F_i=T_i F_i  T_i^{-1}$, $\bar G_{i}=T_iG_i$ and $\bar H_{i}=H_i T_i^{-1}$.
From (\ref{app-sys-2}), it can be easily seen that the triplet $(\bar F_i,\bar G_i,\bar H_i)$ is controllable and observable.
Denote the minimal polynomial of Hurwitz $\bar F_i$ as
$
\mathcal{P}_i(s) = p_{i,0} + p_{i,1}s + \ldots + p_{i,2r_i}s^{2r_i-1} + s^{2r_i}\,.
$
By some straightforward but lengthy calculations, we can deduce that $p_{i,0}=\Pi_{j=1}^{r_i}\gamma_{i,j}^2$.
With this being the case, let $\mathcal{G}(s)$ denote the state transfer function of system (\ref{app-sys-3}), given by
$
\mathcal{G}(s)=\diag(\mathcal{G}_1(s),\ldots,\mathcal{G}_m(s))
$
where $\mathcal{G}_i(s) = {-(\Pi_{j=1}^{r_i}\gamma_{i,j}^2)}/{\mathcal{P}_i(s)}$, with $|\mathcal{G}_i({\infty})|=1<{1/\mu_0}$, $ i=1,\ldots,m
$.

By the Bounded Real Lemma \cite[Theorem 3.1]{Isidori(2017)}, there is a symmetric positive definite  matrix $\bar P_i$ and a  $\bar\lambda_i$ such that
\[
2\chi_i^{\top}\bar P_i(\bar F_i \chi_i + \bar G_{i} u_i)\leq -\bar\lambda_i |\chi_i|^2 + |u_i|^2/\mu_0^2-|y_i|^2
\]
for $i=1,\ldots,m$.
This then suggests
\[
\sum_{i=1}^m2\chi_i^{\top}\bar P_i(\bar F_i \chi_i + \bar G_{i} u_i)\leq -\sum_{i=1}^m\bar\lambda_i |\chi_i|^2 + |u|^2/\mu_0^2-|y|^2\,.
\]
Since $|u|=|\Delta_0 y|\leq |\Delta_0| |y|\leq \mu_0|y|$, we have
\[
\sum_{i=1}^m2\chi_i^{\top}\bar P_i(\bar F_i \chi_i + \bar G_{i} u_i)\leq -\sum_{i=1}^m\bar\lambda_i |\chi_i|^2\,.
\]
Thus, letting $P_i=T_i\bar P_iT_i$ and $\lambda_i\leq \bar\lambda_i|T_i|^2$ for $i=1,\ldots,m$, we complete the proof.

\vspace{-0.5em}
\subsection{Proof of Lemma \ref{lemma-5}}
\label{app-sec-lem5}

Let $\Psi_m= \frac{\lambda_m}{2}\ell_m^2 -\varrho_0 -\kappa \ell_m$ and
\[\ba{l}
\Psi_i=\frac{\lambda_i}{4}\ell_i^2-\sum_{j=i+1}^m \frac{2(j-1)|P_j|^2\iota_{ji}^2}{\lambda_j}\ell_j^{2(r_j-r_i+1)}-\varrho_0 -\kappa \ell_i\,
\ea\]
for $1\leq i\leq m-1$, which indicates that the proof is completed if it is shown that $\Psi_i\geq0$ for all $1\leq i\leq m$.
We proceed to show this by a recursive method.

\emph{{\bf Step} 1:} Let us consider the case that $i=m$. With the choice of $\ell_m$ given in (\ref{c-ell}), choosing $g_m>{2}/{\lambda_m}$ and letting $\mu_m={\lambda_m}g_m^2/2-g_m$, we observe that $\mu_m>0$. Thus, it can be seen that $\Psi_m\geq 0$ for all $\kappa \geq \theta_m$ with $\theta_m = \max\{1,\sqrt{{\varrho_0}/{\mu_m}}\}$.

\emph{{\bf Step} 2:} With the choice of $\ell_{m-1}$ in (\ref{c-ell}),  $\Psi_{m-1}$ reads as
\[\ba{l}
\Psi_{m-1} = \left[\frac{\lambda_{m-1}}{4}g_{m-1}^2 - \frac{2(m-1)|P_m|^2\iota_{m,m-1}^2}{\lambda_m}\right]\ell_m^{2(r_m-r_{m-1}+1)} \\ \qquad -\varrho_0 -\kappa g_{m-1}\ell_{m}^{(r_m-r_{m-1}+1)}\,\\
\geq \mu_{m-1}\kappa^{2(r_m-r_{m-1}+1)} -\varrho_0
\ea\]
where the inequality is obtained using $\kappa\geq 1$ and defining
\[\ba{l}
\mu_{m-1}:=\left[\frac{\lambda_{m-1}}{4}g_{m-1}^2 - \frac{2(m-1)|P_m|^2\iota_{m,m-1}^2}{\lambda_m}\right] g_m^{2(r_m-r_{m-1}+1)}\,\\ \qquad -g_{m-1}g_m^{(r_m-r_{m-1}+1)}\,.
\ea\]
Given any fixed $g_m$, it is clear that there exists a positive constant $g_{m-1}^\ast>0$, independent on $\kappa$ such that $\mu_{m-1}>0$ for all  $g_{m-1}>g_{m-1}^\ast$.
This further implies $\Psi_{m-1}\geq0$ for all $\kappa\geq \theta_{m-1}:=
\theta_{m-1}= \max\left\{1,({\varrho_0}/{\mu_{m-1}})^{\frac{1}{2(r_m-r_{m-1}+1)}}\right\}.
$

\emph{{\bf Step} m-i+1:} Following the previous design, we now proceed to the $m-i+1$-th step, $i=1,\ldots,m$, and have fixed $g_j$ and $\theta_j$ for $j=i+1,\ldots,m$. With (\ref{c-ell}), we observe that
{
\[\ba{l}
\ell_i = g_i(\ell_{i+1})^{r_{i+1}-r_i+1}\,\\
 = g_i(g_{i+1})^{r_{i+1}-r_i+1}(\ell_{i+2})^{\Pi_{k=1}^2(r_{i+k}-r_{i+k-1}+1)}\,\\
 = g_i \Pi_{q=i+1}^{m-1}(g_{q})^{\Pi_{k=i+1}^{m-1}(r_{k}-r_{k-1}+1)} (\ell_{m})^{\Pi_{k=i+1}^{m}(r_{k}-r_{k-1}+1)}\\
 = g_i \Pi_{q=i+1}^{m}(g_{q})^{\Pi_{k=i+1}^{m}(r_{k}-r_{k-1}+1)}\kappa^{\Pi_{k=i+1}^{m}(r_{k}-r_{k-1}+1)}.
\ea\]
Then  $\Psi_{i}$ reads as
\beeq{\label{poly}\ba{l}
\Psi_i(\kappa)=
\dst\omega_{i,i} \kappa^{\dst2\Pi_{k=i+1}^{m}(r_{k}-r_{k-1}+1)} \,\\  - \dst\sum_{j=i+1}^m \omega_{i,j} \kappa^{\dst2(r_j-r_i+1)\Pi_{k=j+1}^{m}(r_{k}-r_{j-1}+1)} \,\\  - \omega_{i,0} \kappa^{\dst\Pi_{k=i+1}^{m}(r_{k}-r_{k-1}+1)+1} -\varrho_0
\ea}
where
\beeq{\label{omega}\ba{l}
\omega_{i,i} = ({\lambda_i}/{4})g_i^2\dst\Pi_{q=i+1}^{m}(g_{q})^{2\Pi_{k=i+1}^{m}(r_{k}-r_{k-1}+1)}\,\\
\omega_{i,j} = \dst\frac{2(j-1)|P_j|^2\iota_{ji}^2}{\lambda_{j}}(g_j)^{2(r_j-r_i+1)}\cdot\,\\ \qquad\cdot \dst\Pi_{q=j+1}^{m} (g_q)^{2(r_j-r_i+1)\Pi_{k=j+1}^{m}(r_k-r_{k-1}+1)}\,,\\ \qquad j=i+1,\ldots,m\,\\
\omega_{i,0} = g_i\dst\Pi_{q=i+1}^{m}(g_q)^{\Pi_{k=i+1}^{m}(r_k-r_{k-1}+1)}\,.
\ea}
}
In this way, the $\Psi_i$ in (\ref{poly}) is expressed as a polynomial of $\kappa$.
Moreover, it is noted that $r_i\leq r_{i+1}\leq\cdots\leq r_m$, and
\[\ba{l}
\Pi_{k=i+1}^{m}(r_{k}-r_{k-1}+1) \\ \geq (r_j-r_i+1)\Pi_{k=j+1}^{m}(r_{k}-r_{k-1}+1)\geq 1
\ea\]
holds for all $j=i+1,\ldots,m$. Thus, given $\kappa\geq1$ we have
\[
\Psi_i \geq \mu_i\kappa^{\dst2\Pi_{k=i+1}^{m}(r_{k}-r_{k-1}+1)}-\varrho_0\,.
\]
with $\mu_i:=\omega_{i,i}-\sum_{j=i+1}^m\omega_{i,j}-\omega_{i,0}$.
Recalling (\ref{omega}), it can be seen that given any fixed $g_j$, $j=i+1,\ldots,m$, there always exists $g_i^\ast>0$, independent on $\kappa$ such that $\mu_i >0$ for all $g_i>g_i^\ast$.
With the above choice of $g_i$ being the case, it then can be easily shown that there exists a $\theta_i>0$ such that for all $\kappa\geq\theta_i$, the polynomial function $\Psi_i$ is positive.

Finally, following the previous recursive design, at the step $m$ we can fix $g_1$ and $\theta_1$. Therefore, choosing $\theta^\ast=\max\{\theta_1,\ldots,\theta_m\}$, we can conclude that for any $\kappa\geq\theta^\ast$, $\Psi_i\geq0$ hold for all $i=1,\ldots,m$, which completes the proof.

\end{document}